\documentclass{llncs}
\usepackage{amssymb,graphics,amsmath, latexsym}
\usepackage[usenames]{color, epsfig}
\usepackage{subfigure}
\usepackage{hyperref}

\newcommand{\eps}{\epsilon}

\newtheorem{defin}[theorem]{Definition}

\newtheorem{fact}[theorem]{Fact}

\makeatletter
\def\imod#1{\allowbreak\mkern10mu({\operator@font mod}\,\,#1)}

\renewcommand{\P}{\mathop{\mathrm{P}}}

\begin{document}
\pagestyle{headings}
\title{Approximate Counting of Matchings in Sparse  Hypergraphs}
\author{
  Marek Karpi\'{n}ski\thanks{Research supported partly by DFG grants
      and the Hausdorff Center grant EXC59-1.}\inst{1} and   Andrzej Ruci\'{n}ski\thanks{Research supported by  the Polish NSC grant N201 604 940.}\inst{2} and Edyta Szyma\'{n}ska\thanks{Research supported by the Polish NSC grant N206 565
740}.\inst{2}}

\institute{Department of Computer Science, University of Bonn.
    Email: marek@cs.uni-bonn.de\and
Faculty of Mathematics and Computer Science,
    Adam Mickiewicz University, Pozna\'{n}. Email: rucinski,edka@amu.edu.pl}

\date{}
\maketitle

\begin{abstract}
In this paper we give a fully polynomial randomized approximation scheme (FPRAS) for the number of
all matchings in  hypergraphs belonging to a class of sparse, uniform hypergraphs. Our method is
based on a generalization of the canonical path method to the case of uniform hypergraphs.

\end{abstract}


\section{Introduction}\label{intro}

 \emph{A hypergraph} $H=(V,E)$ is a finite set of vertices $V$ together with a family $E$ of
distinct, nonempty subsets of vertices called edges. In this paper we consider $k$-\emph{uniform
hypergraphs} (called also \emph{$k$-graphs}) in which, for a fixed $k\ge2$,  each edge is of size $k$.   \emph{A matching} in a hypergraph
is a set of disjoint edges.
We will often identify a matching $M$ with the hypergraph $H[M]$ induced by $M$ in $H$.
\emph{The intersection graph}  of a hypergraph $H$ is a graph $L:=L(H)$ with vertex set $V(L)=E(H)$ and the edges set $E(L)$ consisting of all intersecting pairs of edges of $H$. When $H$ is a graph, the intersection graph $L(H)$ is called \emph{the line graph} of $H$.

In a seminal paper \cite{js}, Jerrum and Sinclair constructed an FPRAS (see Section \ref{1.1} for
the definition) for counting the number of all matchings in graphs (a.k.a. the monomer-dimer
problem) based on an ingenious technique of canonical paths.
The method was extended later in \cite{JSV} to solve the permanent problem.
Here we  test for what classes of
hypergraphs their method can be carried over. It turns out that, given a $k$-graph $H$, we can
adopt the proof for graphs only when for  every two matchings $M,M'$ in $H$ the intersection graph
$L=L(M\cup M')$ between $M$ and $M'$ has maximum degree $\Delta(L)\leq 2.$ This happens if and only
if $H$ contains no \emph{3-comb}, i.e. a $k$-graph $H_0$ consisting of  a matching
$\{e_1,e_2,e_3\}$  and one extra edge $e_4$  such that $|e_4\cap e_i|\ge1$ for $i=1,2,3$ (see Fig.\ref{3comb}). Let us denote by ${\cal H}_0$ the family of all $k$-graphs which do not contain $H_0$. Our main
result is the following hypergraph generalization of the Jerrum-Sinclair theorem. In fact, they
considered the weighted case, while we, for clarity, assume that the hypergraphs are unweighted.
However, the weighted case can be handled in a similar manner.
%
%
\begin{theorem}\label{main}
There exists an FPRAS for  the problem of counting all matchings in $k$-graphs $H\in{\cal H}_0$.
\end{theorem}

 Theorem \ref{main} is proved in Section \ref{markov}. In Section \ref{examples} we give several examples of  classes of $k$-graphs which belong
to ${\cal H}_0$. Finally, in Section \ref{fur} we discuss obstacles preventing us from extending our result to all hypergraphs.

 We can characterize family ${\cal H}_0$ in terms of the intersection graph
$L(H)$. Namely, a $k$-graph $H\in{\cal H}_0$ if and only if the intersection graph $L(H)$ of $H$ is
\emph{claw-free}, that is $L(H)$ does not contain an induced subgraph isomorphic to the star $K_{1,3}$.
For $k=2$ every $k$-graph, i.e., every graph is in ${\cal H}_0$. For $k\ge3$, the requirement that $H\in{\cal H}_0$ is quite restrictive and causes the hypergraph to be rather
sparse (of size $O(n^{k-1})$). Nevertheless, as can be seen in the next subsection, the problem of counting matchings in $k$-graphs belonging to ${\cal H}_0$ is still quite hard.

\subsection{$\sharp$P-Hardness}

We prove in this section that the problem of counting matchings in $k$-graphs belonging to the family ${\cal H}_0$ is $\sharp$P-complete. Let $\sharp$MATCH$_0^k$ denote this problem.

\begin{lemma}\label{hardness}
The problem $\sharp$MATCH$_0^k$ is $\sharp$P-complete for every $k\ge3$.
\end{lemma}
\begin{proof} We use a reduction from the problem of counting all matchings in bipartite graphs $G=(V,E)$ of  maximum degree at most four, which, by a result of Vadhan \cite{vadhan}   is $\sharp$P-complete. For a given bipartite graph $G=(V,E)$ of maximum degree at most four with a bipartition $V=V_1\cup V_2$ we construct a $k$-graph $H=(V',E')$ from the family ${\cal H}_0$ as follows. For every edge $e\in E$ we add to $V$ additional $k-2$ vertices, so $V'=V\cup \bigcup_{e\in E}\{v^e_1,v^e_2,\ldots,v^e_{k-2}\}.$ Now, every edge $e=(u,v)\in E$ is replaced by the corresponding $k$-tuple $(v,v^e_1,v^e_2,\ldots,v^e_{k-2},u).$ Thus $|V'|=|V|+(k-2)|E|$, $|E|=|E'|$ and the resulting $k$-graph $H'=(V',E')$ is simple, $k$-partite, has maximum vertex degree at most four and, more importantly, does not contain any 3-comb. Moreover, there is a natural one-to-one correspondence between the matchings in $G$ and the matchings in $H.$
\qed\end{proof}

\subsection{Motivation from statistical physics}

In 1972 Heilmann and Lieb \cite{HL} studied monomer-dimer systems, which in the graph theoretic
language correspond to (weighted) matchings in  graphs. In physical applications these graphs are
typically some (infinite) regular lattices. Dimers represent diatomic molecules which occupy disjoint pairs of adjacent vertices of the lattice and monomers are the remaining vertices.
Heilmann and Lieb proved that the associated Gibbs measure is unique (in other words, there is no phase transition). They did it by proving that the roots of the generating matching polynomial of any  graph are all real, equivalently that the roots of the hard core partition function (independence polynomial) of any line graph  are all real. The latter result was later extended to all claw-free graphs by Chudnovsky and Seymour \cite{CS}. The uniqueness of Gibbs measure on $d$-dimensional latticed was reproved in a slightly stronger form and by a completely different method by Van der Berg \cite{VdB}.

 Hypergraphs may be at hand when instead of diatomic molecules
 bigger molecules (polymers) are considered which, again, can
 occupy ``adjacent'', disjoint sets of vertices of a lattice.
 As long as the hypergraph lattice $H$ belongs to the family ${\cal H}_0$,
 the intersection graph $L(H)$ is claw-free (because $H$ contains no 3-comb)
 and, by the result of \cite{CS} combined with the proof from \cite{HL} there is no phase transition either.
 However, it is possible to have a phase transition
 for a monomer-trimer system (cf. \cite{Heil}). Interestingly, the example given
 by Heilmann (the decorated, or subdivided, square lattice with hyperedges
 corresponding to the collinear triples with midpoints at the branching points of the original square lattice)
 is a 3-uniform hypergraph containing  3-combs, and thus its intersection graph is \emph{not} claw-free (see Fig.\ref{siatka}).
\begin{figure}
\centering
\subfigure[a 4-uniform 3-comb]{\scalebox{0.7}{\label{3comb}\input{3comb.pstex_t}}}
\qquad\qquad\qquad
\subfigure[Heilmann's 3-graph lattice; the shaded edges form a 3-comb]{\label{siatka}\includegraphics{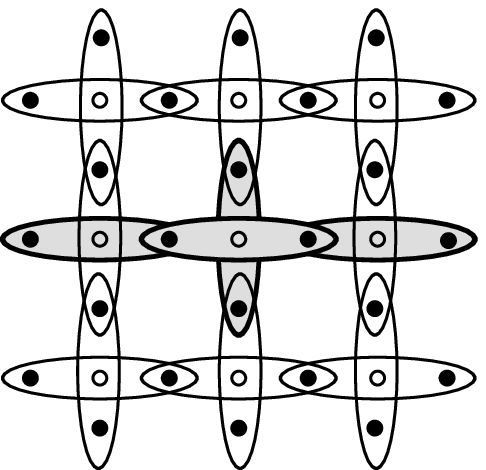}}
\caption{}
\label{fig1}
\end{figure}
%

%
%

\subsection{Related  results}

Recently, an alternative  approach to constructing counting
schemes for graphs has been developed based on the concept of
spatial correlation decay. This resulted in deterministic fully
polynomial time approximation schemes (FPTAS) for counting
independent sets in graphs with degree at most five (\cite{wei}),
proper colorings (\cite{GK}), and matchings in graphs of bounded
degree (\cite{bgknt}). It is not clear to what extent these methods can
be applied to hypergraphs.

The above mentioned FPTAS for counting independent sets in graphs implies an FPTAS for counting
matchings in hypergraphs whose intersection graphs have maximum degree at most five. This is the
case of the hexagon based lattice (see Subsection \ref{kag}) and, with high probability,
of the random hypergraphs discussed in Subsections \ref{rh} and \ref{tr}. More importantly, this is the case also of the
the Heilmann lattice described in the previous
subsection (the maximum degree of its intersection graph is three), which undermines our temptation to link the absence of phase transition for a hypergraph lattice with the absence of a 3-comb, that is with the claw-freeness of the intersection graph of the lattice.

As far as hypergraphs are concerned, the authors of \cite{BDK}
showed that, under certain conditions, the Glauber dynamics for
independent sets in a hypergraph, as well as the Glauber dynamics
for proper q-colorings of a hypergraph mix rapidly. It is
doubtful, however, if the path coupling technique applied there
can be of any use for the problem of counting matchings in
hypergraphs. Nevertheless, paper \cite{BDK} marks a new line of
research, as there have been only few results (\cite{bdgj},
\cite{bd}) on approximate counting in hypergraphs before. The only
other paper devoted to counting matchings in hypergraphs we are
aware of is \cite{barvinok}, where Barvinok and Samorodnitsky
compute the partition function for matchings in hypergraphs under
some restrictions on the weights of edges. In particular they are
able to distinguish in polynomial time between hypergraphs that
have sufficiently many perfect matchings and hypergraphs that do
not have nearly perfect matchings.

\subsection{Approximate Counting and Unform Sampling}\label{1.1}

Given $\eps>0$ and $\delta>0$, we say that a random variable $Y$ is an $(\eps, \delta)$-\emph{approximation} of a constant $C$ if
\[\P\left(|Y-C|\ge\eps C\right)\le\delta.\]
Consider a problem of evaluating  a function $f$ over a set of input strings $\Sigma^*$.

\begin{defin}\label{fpras}\rm
A randomized algorithm  is called a \emph{fully polynomial randomized approximation scheme (FPRAS)
for $f$} if for every triple $(\eps,\delta, x)$ with $\eps>0,\;\delta>0$, and $x\in \Sigma^*,$ the
algorithm returns an  $(\eps, \delta)$-\emph{approximation} $Y$ of $f(x)$ and runs in time
polynomial in  $1/\eps$, $\log(1/\delta)$, and $|x|$.
\end{defin}

Consider a counting problem, that is,  a problem of computing $f(x)=|\Omega(x)|$, where $\Omega(x)$
is a well defined finite set associated with $x$ (think of the set of all matchings in a
hypergraph). As it turns out (see below), to construct an FPRAS for such a problem it is sufficient
to be able to efficiently sample an element of $\Omega(x)$ almost uniformly at random. To make it
precise, given $\eps>0$, we say that a probability distribution $\P:2^\Omega\to[0,1]$ over a finite
sample space $\Omega$ is $\eps$-uniform if for every $S\subseteq \Omega$
$$\left|\P(S)-\frac{|S|}{|\Omega|}\right|\le \eps,$$
that is, if the total variation distance, $d_{TV}(\P,\tfrac1{|\Omega|})$, between the two
distributions is bounded by $\epsilon$.

\begin{defin}\label{fpaus}\rm
A randomized algorithm  is called a \emph{fully polynomial almost uniform sampler (FPAUS)}
 for a counting problem $|\Omega(x)|$ if for every
pair $(\eps, x)$ with $\eps>0$ and $x\in \Sigma^*,$ the algorithm samples $\omega\in\Omega$
according to an $\eps$-uniform distribution $\P$ and runs in time polynomial in  $1/\eps$ and
$|x|$.
\end{defin}

It has been proved by Jerrum, Valiant, and Vazirani \cite{JVV} that for a broad class of  counting
problems, called self-reducible, including the matching problem, knowing an FPAUS allows one to
construct an FPRAS. For a  proof in the graph case see Proposition 3.4 in \cite{jerrum-book}. The hypergraph
case follows mutatis mutandis. Thus, the proof of Theorem \ref{main} reduces to constructing an
FPAUS for matchings in $H$.

In fact, this approach has been invented for matchings in graphs  already by Broder in
\cite{broder}, and successfully executed by Jerrum and Sinclair in \cite{js}. In their version the
main steps of finding an efficient  FPAUS for  matchings in a graph $H$ were

\begin{itemize}
\item[$\bullet$] a construction of an ergodic time-reversible, symmetric Markov chain $\mathcal{MC}(H)$ whose
 state space $\Omega$ consists of
all matchings in $H$;
 \item[$\bullet$] a proof that $\mathcal{MC}(H)$ is rapidly mixing.

\end{itemize}

\subsection{Rapid mixing}
Given an arbitrary probability distribution $\P_0$ on the state space $\Omega$, let us define the
mixing time $t_{mix}(\epsilon)$ of a Markov chain $\mathcal{MC}$ as
$$t_{mix}(\epsilon)= \min\{t: d_{TV}({\P}_t,\tfrac1{|\Omega|})\le \epsilon\},$$
where $\P_t$ is the chain's state distribution after $t$ steps, beginning from the initial
distribution $\P_0$.  Recall that if an ergodic time-reversible Markov chain is symmetric, i.e.,
the transition probabilities satisfy $p_{ij}=p_{ji}$ for all $i,j\in \Omega$, then its unique
stationary distribution is uniform (cf. \cite{jerrum-book}). In that case we define the transition graph
$G_{\mathcal{MC}}=G$ of $\mathcal{MC}$ as a  graph on the vertex set $V(G)=\Omega$ and the edge set
$E(G)=\{\{i,j\}: p_{ij}>0\}$.  Note that $G$ is undirected but, possibly, with loops. The pivotal
role in estimating the rate of convergence of $\mathcal{MC}$ to its uniform stationary distribution
is played by an expansion parameter, called \emph{ the conductance } and denoted
 $\Phi(\mathcal{MC})$ which in the symmetric case is  defined by a simplified formula
\begin{equation}\label{cond}
\Phi:=\Phi(\mathcal{MC})= min_{S\subseteq \Omega:\,0<|S|\leq \tfrac 12|\Omega|}
  \frac{\sum_{ij\in G}p_{ij}}{|S|}.
  \end{equation}
Indeed,
it follows from Theorem 2.2 in \cite{js} that if $p_{ii}\ge\tfrac12$ for all $i\in \Omega$ then
\begin{equation}\label{lam}
d_{TV}({\P}_t,\tfrac1{|\Omega|})\le|\Omega|^2\left(1-\Phi^2/2\right)^t,
\end{equation}
regardless of the initial distribution ${\P}_0$, and consequently,
\begin{equation}\label{mix}
t_{mix}(\epsilon)\le\frac2{\Phi^2}\left(2\log|\Omega|+\log\epsilon^{-1}\right)
\end{equation}
 Hence, it becomes crucial to estimate the
conductance from below by the reciprocal of a polynomial in the input size. To this end, observe that

\begin{equation}\label{p}
\Phi(\mathcal{MC})\ge min_{S\subseteq \Omega:\,0<|S|\leq \tfrac 12|\Omega|}
  \frac{p_{\min}|cut(S)|}{|S|},
  \end{equation}
  where $cut(S)$ is the edge-cut of  $G$ defined by $S$, and
  $$p_{\min}=\min\{p_{ij}:\; \{i,j\}\in G,\; i\neq
  j\}.$$
 For Markov chains on matchings of an $n$-vertex $k$-graph $H$, denoted further by $\mathcal{MC}(H)$, to bound $|cut(S)|$,  Jerrum and Sinclair introduced their  method of
 canonical paths which boils down to:

\begin{itemize}
\item[$\bullet$] defining a \emph{canonical path} in $G$ for every pair of matchings $(I,F)$ in $H$;
\item[$\bullet$] bounding  from above the number of canonical paths containing a prescribed
 transition (an edge of $G$) by $poly(n)|\Omega|$.
\end{itemize}

Since every canonical path between a matching in $S$ and a matching in the complement of $S$ must go through an edge of $cut(S)$, we have, clearly,
\begin{equation}\label{cut}
|cut(S)|\ge\frac{|S|(|\Omega|-|S|)}{poly(n)|\Omega|}\ge\frac{|S|}{2poly(n)}
 \end{equation}
 and, consequently,
\begin{equation}\label{Phi}\Phi(\mathcal{MC}(H))\ge\frac{p_{\min}}{poly(n)}.\end{equation}

Our plan for proving Theorem \ref{main} is to basically  follow the footprints of \cite{js} for as
long as it is feasible for the more complex structure of hypergraphs. Our proof is given in the
next section, while Section \ref {examples} contains examples of $k$-uniform hypergraphs to which
Theorem \ref{main} applies. But first we collect together various definitions of cycles in
hypergraphs, a couple of which  will be used later in the paper.

\subsection{Cycles in hypergraphs}\label{hypcyc}

A hypergraph with edges $\{e_1,\dots,e_m\}$, where $m\ge3$, is \emph{a Berge-cycle} if there is \emph{a subset}
$\{v_1,\dots,v_{m}\}$ of its vertex set (called \emph{the core}) such that for every $i=1,\dots,m$,
the vertices $v_i$ and $v_{i+1}$ belong to $e_i$ where $v_{m+1}:=v_1$.
(See Fig.\ref{Bcyc} for an example of a Berge-cycle).
The structure of a Berge-cycle, in general, may be quite chaotic, not resembling what we think a
cycle  should look like.
\begin{figure}
\centering
\subfigure[a 4-uniform Berge cycle]{\scalebox{0.5}{\label{Bcyc}\input{Bergecyc.pstex_t}}}
\qquad\qquad
\subfigure[a necklace]{\scalebox{0.5}{\label{neck}\input{neclace.pstex_t}}}
\caption[]{}
\label{cycles}
\end{figure}

\emph{A necklace} is a Berge-cycle whose \emph{all} vertices can be ordered  cyclically  in such a way
that every edge forms a segment of this ordering -- see Fig.\ref{neck}. (Then the first vertices of each edge form the
core.) A necklace $C$ is called \emph{loose} if $\Delta(C)\le 2$.  \emph{An open loose necklace} is a hypergraph
with edge set $\{e_1,\dots, e_m\}$, where $m\ge1$, such that for every $1\le i<j\le m$, $e_i\cap e_j\neq\emptyset$ if and only if $j=i+1$.

Moving to $k$-uniform hypergraphs, for  a pair of natural numbers $k\ge 3$ and $1\le \ell\le k-1$,
define an {\em $\ell$-overlapping $k$-cycle} $C_n^{(k)}(\ell)$ as a necklace of $n$ vertices in which every two
consecutive edges (in the natural ordering of the edges induced by the ordering of the vertices)
share exactly $\ell$ vertices. The case of $\ell=k-1$ is referred to as the \emph{tight} cycle.
Note that  $C_n^{(k)}(\ell)$  has precisely $\tfrac{n}{k-\ell}$ edges.  Note also that for $\ell\le
k/2$ the cycle $C_n^{(k)}(\ell)$ is a loose necklace.

\section{The proof of Theorem \ref{main}}\label{markov}
In this section we define a  Markov chain whose states are the matchings of a $k$-uniform
hypergraph $H$ and then prove Theorem \ref{main} by showing that the chain is rapidly mixing to a uniform stationary distribution.

\subsection{The Markov Chain}
Given a $k$-graph $H=(V,E)$, $|V|=n$, let $\Omega(H)$ denote the set of all matchings in $H.$ We define a
Markov  chain $\mathcal{MC}(H)=(X_t)_{t=0}^{\infty}$ with state space $\Omega(H)$ as follows. Set
$X_0=\emptyset$ and for $t\ge 0$, let $X_t$ be a matching $M=\{h_1,h_2,\ldots,h_s\}$ with $h_i\in
H$, $0\le s\le n/k$. Choose an
edge $h\in H$ uniformly at random and consider the set  $S_h:=\{i:\,h\cap h_i\neq
\emptyset,\,i=1,\ldots,s \}$. The following transitions from $X_t$ are allowed in $\mathcal{MC}(H)$:
\begin{itemize}
\item[(-)] if $h\in M$ then $M':=M-h,$

\item[(+)] if $h\notin M$ and $|S_h|= 0$ then
$M':=M+h,$

\item[(+/-)] if $h\notin M$ and $S_h=\{j\}$ then
$M':=M+h-h_j,$

\item[(0)] if $h\notin M$ and $|S_h|\ge2$ then
$M':=M$.

\end{itemize}
Finally, with probability 1/2 set $X_{t+1}:=M'$, else $X_{t+1}:=X_t.$

\begin{fact} The Markov chain $\mathcal{MC}(H)$ is ergodic and symmetric.
\end{fact}
\proof First note that this chain is irreducible
 (one can get from any matching to any other matching by a sequence of above transitions) and
 aperiodic (due to loops), and so it is ergodic.
 To prove the symmetry of $\mathcal{MC}(H)$, note that for two different matchings $M,M'\in\Omega(H)$,
 the transition probability
\begin{equation}\label{p}
{\P}_{M,M'}=\begin{cases}\frac 1{|H|} &\mbox{ if } |M\oplus M'|=1\\
\frac 1{|H|}&\mbox{ if } M\oplus M'=\{e,f\}, e\cap f\neq \emptyset\\
0&\mbox{ otherwise }.\end{cases}\end{equation}
 Thus,
$\P_{M,M'}=\P_{M',M}$ .\qed

\bigskip

The above fact implies that $\mathcal{MC}(H)$ converges to a stationary distribution that is
uniform over $\Omega(H)$. Moreover, it follows from equation (\ref{p}) that
\begin{equation}\label{pmin}
p_{\min}=\min\{P_{M,M'}:\; \{M,M'\}\in G,\; M\neq
  M'\}=\frac1{|H|}\ge n^{-k}.
  \end{equation}

\subsection{Canonical Paths}

In this section we will define canonical paths, a tool  used for
estimating the mixing time of the Markov chain $\mathcal{MC}(H)$
introduced in the previous subsection. For brevity,  loose
necklaces will be called \emph{cycles}, while open loose necklaces
will be called \emph{paths}. (Note that a pair of edges sharing at
least two  vertices is a path, not a cycle.)

Set  $V(H)=\{1,2,\ldots,n\}$ and $\min S=\min\{i: i\in S\}$ for any $S\subseteq V(H)$.
Let $(I,F)$ be an ordered pair of matchings in $\Omega(H)$ (we might think
of them as the initial and the final matching of the canonical path-to-be).  The symmetric
difference $I\oplus F$  is a hypergraph with $\Delta(I\oplus F)\le2$ and, due to the
assumption that $H\in{\cal H}_0$, also $\Delta(L(I\oplus F))\le2$, that is, in $I\oplus F$ every edge intersects at most two other edges.
 Hence, each component of $I\oplus F$ is a path or a cycle, in which
  the edges of $I$ alternate with the edges $F$. In particular, each cycle-component has an even number of edges.

Let us order the components $Q_1,\dots,Q_q$ of $I\oplus F$ so that $\min V(Q_1)<\cdots<\min V(Q_q)$. We construct the canonical path $\gamma(I,F)=
(M_0,\dots,M_t)$ in the transition graph $G$ by setting $M_0=I$ and then modifying the current
matching by transitions (+), (-), or (+/-), while
 traversing  the components $Q_1,\dots, Q_q$ as follows. For the sake of uniqueness of the canonical path, each component will be traversed from a well defined starting point (an edge $e_1$) and in a well defined direction $e_1,e_2,\dots e_s$. Of, course, for a path there are just two starting points (which determine directions), while for a cycle there are $s$ starting points and two directions from each. The particular rules for choosing the starting point and direction are quite arbitrary and do not really matter for us.  Suppose that we have
already constructed matchings $M_0,M_1,\dots, M_{j}$ and traversed so far the components
$Q_1,\dots,Q_{r-1}$.

If $Q_{r}$ is an even path then we assume that $e_1\in F$ (and so $e_s\in I$) and take
$M_{j+1}=M_j+e_1-e_2$, $M_{j+2}=M_{j+1}+e_3-e_4$,..., $M_{j+s/2}=M_{j+s/2-1}+e_{s-1}-e_s$.
If $Q_{r}$ is an odd path then we assume that $\min(e_1\cap e_2)<\min(e_{s-1}\cap e_s)$. If
 $e_1, e_s\in I$ then take $M_{j+1}=M_j-e_1$,
$M_{j+2}=M_{j+1}+e_2-e_3$, $M_{j+3}=M_{j+2}+e_4-e_5$, ...,  $M_{j+(s+1)/2}=M_{j+(s-1)/2}+e_{s-1}-e_s$. If $e_1,e_s\in F$, we apply the
sequence of transitions $M_{j+1}=M_j+e_1-e_2$, $M_{j+2}=M_{j+1}+e_3-e_4$,...,$M_{j+(s-1)/2}=M_{j+(s-3)/2}+e_{s-2}-e_{s-1}$, and $M_{j+(s+1)/2}=M_{j+(s-1)/2}+e_s$.
Finally, if $Q_{r}=(e_1,\dots,e_s)$ is a cycle then we assume that $\min e_1=\min(V(Q_r)\cap V(I))$ and $\min(e_2\cap e_3)>\min(e_{s-1}\cap e_s)$, and
follow the sequence of transitions $M_{j+1}=M_j-e_1$,
$M_{j+2}=M_{j+1}+e_2-e_3$, $M_{j+3}=M_{j+2}+e_4-e_5$, ...,$M_{j+s/2}=M_{j+s/2-1}+e_{s-2}-e_{s-1}$, and
$M_{j+s/2+1}=M_{j+s/2}+e_s$.

We call the component $Q_r$ of $I\oplus F$ \emph{the venue} of the transition $(M_j,M_{j+1})$ (on the canonical path $\gamma(I,F)$) if $M_j\oplus M_{j+1}\subseteq E(Q_r)$.
Note that the obtained sequence $\gamma(I,F)=(M_0,\dots,M_t)$ is unique and satisfies the following properties:
\begin{enumerate}
\item[(a)]  $M_0=I$ and  $M_t=F$,
\item[(b)]   for every $j=0,\dots,t-1$, the pair $\{M_{j},M_{j+1}\}$ is an edge
of the transition graph $G_{\mathcal{MC}(H)}$,
\item[(c)]  for every $j=0,\dots,t$, we have $I\cap F\subseteq M_j\subseteq I\cup F$,
\item[(d)] for every $j=0,\dots,t$, we have $F\cap \bigcup_{i=1}^{r-1}Q_i\subseteq M_j$ and $I\cap \bigcup_{i=r+1}^qQ_i\subseteq M_j$, where $Q_r$ is the venue of $(M_j,M_{j+1})$.

\end{enumerate}

\subsection{Bounding the cuts}

Fix a transition edge $(M,M')$ in  $G_{\mathcal{MC}(H)}$. Let
$\Pi_{M,M'}=\{(I,F):\,(M,M')\in\gamma(I,F)\}$ be the set of canonical
paths containing the transition edge $(M,M').$ Our goal is to show that
\begin{equation}\label{Pi}
|\Pi_{M,M'}|\leq
|\Omega'(H)|,
\end{equation}  where
\[\Omega'(H)=\{H'\subseteq H:\,\exists e\in H'\mbox{ such that }H'-e\in \Omega(H)\}.\]
Note that
$|\Omega'(H)|\le|\{(M,e)\;:\; M\in \Omega(H),\; e\in H\}|\le  n^k|\Omega(H)|$ and $\log|\Omega(H)|=O(n\log n)$.
Thus, in view of the remarks at the end of Section \ref{intro}, the estimates (\ref{mix}), (\ref{cut}), (\ref{Phi}), (\ref{pmin}), and (\ref{Pi}) yield  a polynomial bound on $t_{mix}(\epsilon)$ and thus complete the proof of Theorem \ref{main}.

We will prove (\ref{Pi}) by defining a function $\eta_{M,M'}:\Pi_{M,M'}\to \Omega'(H)$ and  showing that $\eta_{M,M'}$ is an injection.
Fix $(I,F)\in \Pi_{M,M'}$ and define
\begin{equation}\label{eta}
\eta_{M,M'}(I,F)=(I\oplus F)\oplus(M\cup M').
\end{equation}

\begin{fact} For all $I,F\in\Pi_{M,M'}$ we have $\eta_{M.M'}(I,F)\in \Omega'.$
\end{fact}

\proof If $(I,F)\in \Pi_{M,M'}$ then the canonical path $\gamma(I,F)=(M_0,\dots, M_t)$ contains a consecutive pair $M_j=M$ and $M_{j+1}=M'$ for some $j\in\{0,\dots,t\}$. Let $Q_r$ be the component of $I\oplus F$ which is the venue of $(M,M')$ on $\gamma(I,F)$. By the construction of $\gamma(I,F)$ it follows that $\eta_{M,M'}$ is a matching, unless  $Q_r$  is a cycle $(e_1,\dots,e_s)$ and $M'=M+e_\ell-e_{\ell+1}$ for some $\ell\in\{2,4,\dots,s-2\}$. But then, by property (d) above, we  have
$$\eta_{M,M'}=I\cap \bigcup_{i=1}^{r-1}Q_i\cup F \bigcup_{i=r+1}^qQ_i\cup \{e_1,e_3,\dots, e_{\ell-1}, e_{\ell+2},\dots, e_s\}.
$$
Hence, $\eta_{M,M'}-e_1\in\Omega(H)$, and, consequently, $\eta_{M.M'}\in\Omega'(H)$.
 \qed

\begin{fact}\label{injection} The mapping $\eta_{M,M'}:\Pi_{M,M'}\to \Omega'(H)$ is injective.
\end{fact}

\proof We will prove  this fact by showing that any value $\eta$ of this function uniquely determines the pair $(I,F)$ for which $\eta_{M,M'}(I,F)=\eta$. Given $\eta_{M,M'}(I,F)$  we can recover $I\oplus F$ by reversing equation (\ref{eta}):
$$I\oplus F=\eta\oplus(M\cup M').$$
By property (c), we immediately have $I\cap F=M\setminus(I\oplus F)$. It remains to distinguish between the edges of $I\oplus F$ which belong to $I$ and to $F$.
First observe that we can  recover the original ordering of the components $Q_1,\dots, Q_q$ of $I\oplus F$ (by computing $\min V(Q_i)$ for all $i$), as well as the venue $Q_r$ of the transition $(M,M')$ on the canonical path $\gamma(I,F)$ (by locating $M\oplus M'$). By property (d), for every $i<r$ we have $Q_i\cap M\subseteq F$, while for every $i>r$ we have $Q_i\cap M\subseteq I$. To reconstruct $I$ and $F$ on $Q_r$, note that it suffices to identify just one edge of $Q_r$ and then follow the alternating pattern of $I$ and $F$ on $Q_r$. To this end, note that $|M\setminus M'|\le1$ and $|M'\setminus M|\le1$ but $|M\oplus M'|\ge1$. If $M\setminus M'=\{e\}$ then $e\in I$. If $M\setminus M'=\emptyset$ then the unique edge which belongs to $M'\setminus M$ is in $F$.
\qed

\section{ Hypergraphs with no 3-combs}\label{examples}

In this section we give a couple of examples of classes of $k$-uniform hypergraphs which belong to
family ${\cal H}_0$.  A
hypergraph is called \emph{simple} (a.k.a. linear)  when no two edges share two vertices, that is,
the maximum pair degree is one. The proofs of the facts stated in this section will be given in the journal version of the paper.

\subsection{Hypergraphs based on hexagonal lattices}\label{kag}

The kagome lattice is the line graph of the hexagonal lattice. We
construct a 3-graph $H$ by replacing each edge $e=uv$ of the
hexagonal lattice $G$  with a triple $uw_ev$, where $w_e$, $e\in
G$, are new and distinct vertices. Then, the intersection graph
$L(H)$  is still the kagome lattice and thus, $H$ contains no
3-comb. Note, however, that $\Delta(L(H))=4$ and, in view of the
results in \cite{lv}, \cite{wei}, our Theorem~\ref{main} is not
new for such $H$.

\subsection{Random hypergraphs}\label{rh}

Consider a random binomial $k$-graph $H=H^{(k)}(n,p)$ where each $k$-tuple of vertices becomes an
edge, independently, with probability $p=p(n)$. To ensure the absence of copies of the 3-comb $H_0$
in $H$ we need $p=o(n^{-k+3/4})$. Indeed, then the expected number of 3-combs is $O(n^{4k-3}p^4)=o(1). $
Consequently, almost all $k$-graphs on $n$ vertices and with $m=o(n^{3/4})$ edges are $H_0$-free.
Such $k$-graphs, however, are very sparse. For instance,  typically  the maximum vertex degree is
three and the maximum pair degree is one, that is, they are simple.

\subsection{3-graphs on triangles of  random graphs}\label{tr}

A \emph{windmill} is a graph consisting of four triangles:  one central triangle and three other,
mutually disjoint triangles, each of which shares one vertex with the central triangle
(see Fig.\ref{windm}).
\begin{figure}
\centering
\subfigure[a windmill]{\label{windm}\includegraphics[scale=0.3]{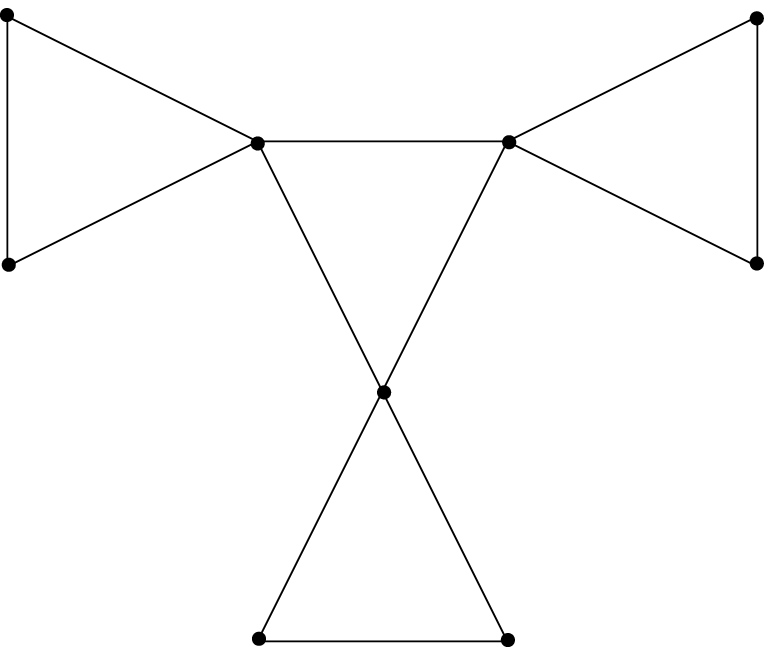}}
\qquad\qquad\qquad
\subfigure[a subdivided 3-graph]{\label{subdiv}\includegraphics[scale=0.2]{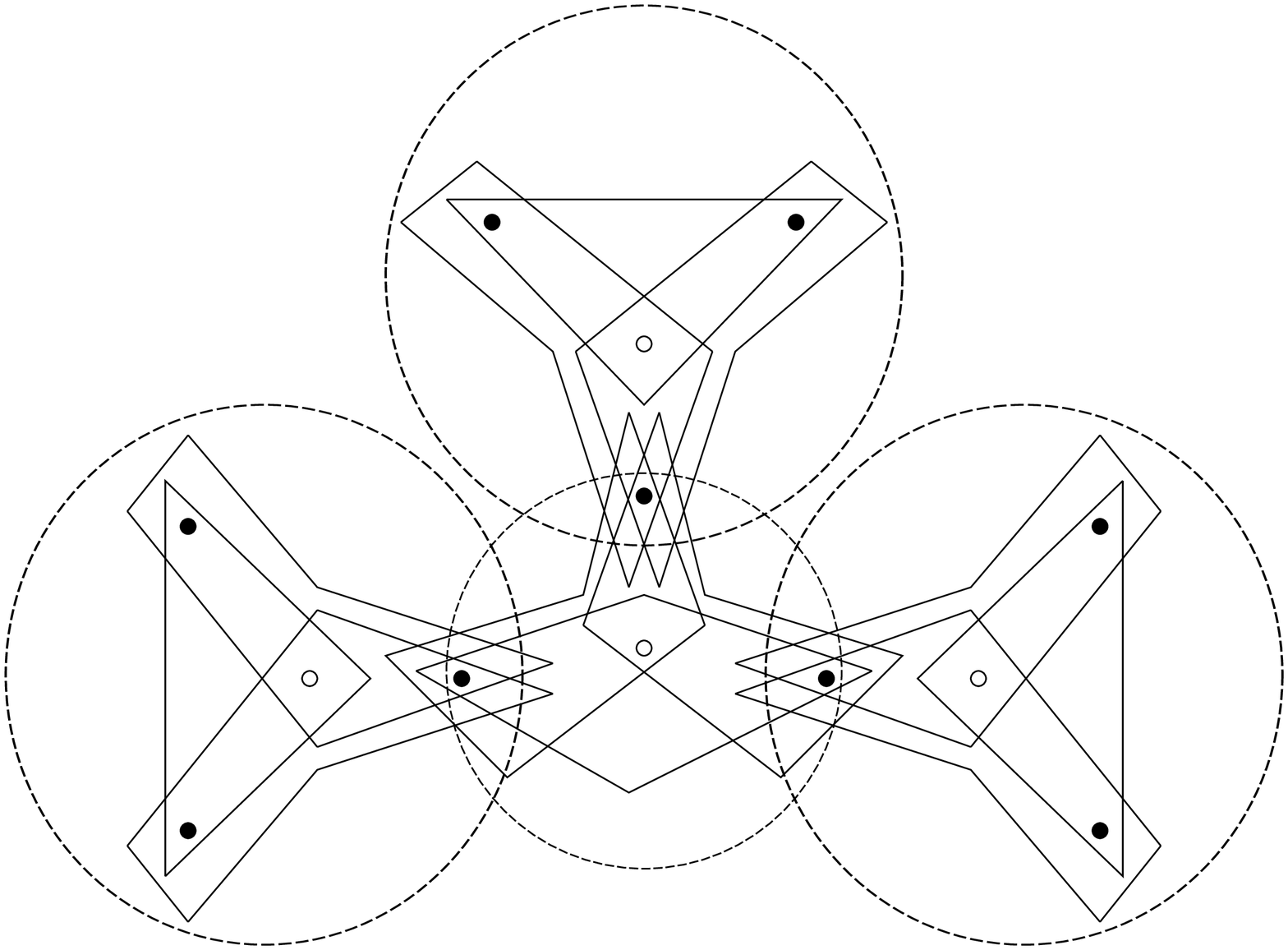}}
\caption{}
\label{fig3}
\end{figure}
 For a graph $G$ define the \emph{triangle 3-graph} $H=T(G)$
 where the hyperedges correspond to the vertex sets of the triangles of $G$.
  Clearly, if $G$ is windmill-free then $H\in{\cal H}_0$.
 The matchings of $H$  correspond
  to  $K_3$-matchings of $G$, i.e. to vertex-disjoint unions of triangles in $G$.

 Note that almost all graphs with $n$ vertices and $m=o(n^{5/4})$ edges are windmill-free. Indeed, a random graph $G(n,p),\,p=o(n^{-3/4}),$ a.a.s. has no windmills because the expected number of windmills is $O(n^9p^{12})=o(1)$.
To have a better insight into the structure of such graphs, observe that for such $p$, a typical
$G(n,p)$ has $o(n^{5/4})$ edges, $o({n^{3/4}})$ triangles and $o(n^4p^5)=o(n^{1/4})$ pairs of
triangles sharing an edge.

\subsection{Subdivided 3-graphs}\label{sub}

 For an  \emph{arbitrary} 3-graph $H=(V,E)$ and a sequence of natural numbers $\nu=(\nu_e: e\in E)$, construct a \emph{subdivided} 3-graph $H'_\nu=(V',E')$ in the following way.
The vertex set is $V'=V\cup V_E$, where $V_E=\bigcup_{e\in H} V_e$, $|V_e|=\nu_e$, and the sets
$V_e$ are pairwise disjoint and disjoint from $V$.
 The edge set $E'$ is obtained by replacing each hyperedge $e=\{v_1,v_2,v_3\}$ with all triples of the form
 $\{v_i,v_j,v\},$ where $1\le i<j\le 3$ and $v\in V_e$.

\begin{fact}
For every $H$ and $\nu$, the hypergraph $H'_\nu$  does not contain any 3-comb, i.e. $H'_\nu\in{\cal
H}_0$.
\end{fact}

\proof Consider an arbitrary edge $f=\{v_i,v_j,v\}$ of $H'_\nu$, where $v\in V_e$ and
$e=\{v_1,v_2,v_3\}$. Every edge which intersects $f$ at $v$ has also another vertex common with
$f$. Thus, $f$ cannot intersect three disjoint edges of $H'_\nu$, and so there is no 3-comb in
$H'_\nu$. \qed

\bigskip

Note that for a simple $H$, every matching $M=\{\{u_1,v_1\},\dots, \{u_t,v_t\}\}$ in the shadow
graph $\Gamma(H)$ of $H$ (obtained by replacing each hyperedge with a graph triangle) determines precisely $\prod_{i=1}^t\nu_{e_i}$ matchings of size $t$ in
$H'_\nu$, where $e_i$ is the unique edge of $H$ which contains the pair $\{u_i,v_i\}$. Moreover,
every matching of $H'_\nu$ is determined this way. Thus,  for simple $H$, the problem of counting
matchings in $H'_\nu$ reduces to counting matchings in graphs with weighted edges.

 In the special case when for all $e\in H$ we have $\nu_e=1$ (see Fig.\ref{subdiv}), the above defined operation generalizes the operation of edge subdivision
for graphs and, as for graphs, it preserves hypergraph planarity. For instance, consider a planar triangulation $G$ and its triangle 3-graph
$H=T(G)$. Clearly, $H$ may contain lots of 3-combs. However, the subdivision $H'_1$ of $H$ is free
of 3-combs, but still represents a triangulation.

\subsection{Enriched tight cycles}\label{enr}

Recall the definition of a $(k-1)$-overlapping $k$-cycle $C_n^{(k)}(k-1)$ (a tight cycle) from
Section \ref{hypcyc} and observe that it is $H_0$-free and has precisely
  $n$ edges.
We will now ``enrich'' this tight cycle by increasing its number of edges by a factor of
$\tfrac{k+1}2$, still keeping it in the family ${\cal H}_0$.

 Fox $k\ge 3$,  let $C=C_n^{(k+1)}(k-1)$ be a $(k-1)$-overlapping $(k+1)$-cycle on
$n$ vertices. We build a $k$-graph $D=D_n^{(k)}$ by replacing each edge $e$ of $C$ by the
$k$-clique $\binom ek$. Note that $D_n^{(k)}$ contains $C_n^{(k)}(k-1)$ and has $\tfrac{k+1}2n$
edges, so, in a sense, it can be viewed as
 an enriched tight cycle.

\begin{fact} For all $k\ge 3$,
the $k$-graph $D_n^{(k)}$ is $H_0$-free.
\end{fact}

\proof Suppose that there is a copy of the 3-comb $H_0$ in $D_n^{(k)}$. Let us denote its pairwise
disjoint  edges by $e_1,e_2,e_3$, and the intersecting edge by $e_4$. Suppose that $e_4$ is
contained in an edge $e$ of $C$. There are only two other edges of $C$ which intersect $e$ and the
total number of vertices in the edges of $C$ intersecting $e$ (including $e$ itself) is at most
$3k-1$. The edges $e_1,e_2,e_3$ must each belong to one of these edges of $C$. But this contradicts
the fact that $e_1,e_2,e_3$ are mutually disjoint. \qed

\subsection{Rooted blow-up hypergraphs}\label{root}

The families from subsections \ref{rh} and \ref{tr} consisted of $k$-graphs with only $o(|V|^{3/4})$
edges, while the $H_0$-free $k$-graphs constructed in subsections \ref{kag}, \ref{sub} and \ref{enr} had
$\Theta(|V|)$ edges. Here, we define $H_0$-free $k$-graphs with up to $\Theta(|V|^2)$ edges.

Partition an $N$-vertex set $V$ into $n$ nonempty sets $V_1,\dots,V_n$, and fix one vertex $v_i \in
V_i$ for each $i=1,\dots,n$. Next, for every pair $1\le i<j\le n$ include to the edge set $E$ the
family $E_{ij}$ of all $k$-element subsets of $V_i\cup V_j$ containing both, $v_i$ and $v_j$.
\begin{fact}
The $k$-graph $D=(V,E)$ is $H_0$-free.
\end{fact}

\proof Suppose that there is a copy of the 3-comb $H_0$ in $H$. Let us denote its pairwise edges by
$e_1,e_2,e_3$ and the intersecting edge by $e_4$. There exists a unique pair $1\le i<j\le n$ such
that $e_4\in E_{ij}$. Then each of $e_1,e_2,e_3$ intersects the set $V_i\cup V_j$, and consequently
contains either $v_i$ or $v_j$.  But this contradicts the fact that $e_1,e_2,e_3$ are mutually
disjoint. \qed

\bigskip

 Note that when $|V_i|=O(1)$ for all $i$, the hypergraph $D$ has $\Theta(n^2)$ edges. In fact,  when $|V_i|=1$ and $k=2$ we obtain the complete graph $K_n$.

 \section{Further Research}\label{fur}

    It remains an open question how to extend
     our result to larger classes of hypergraphs. The success  in the case of graphs relied mostly on the fact that every graph is free of 3-combs and thus $I\oplus F$ has a very simple structure. This is the case of the hypergraphs in the family ${\cal H}_0$ as well. For general hypergraphs,
     however, the presence of 3-combs may lead to situations where $e_1,e_2,e_3\in I$, $e_4\in F$,  and $e_4\cap e_i\neq\emptyset$, $i=1,2,3$. Then, in the  process of creating the canonical path $\gamma(I,F)$, even if we could somehow traverse the component $Q$ of $I\oplus F$ containing the 3-comb $e_1,e_2,e_3,e_4$, in order to put $e_4$ on the current matching $M_j$ we would need first to delete $e_1$ and $e_2$, and at least one of them, say $e_2$, by a transition of type (-). As $e_2$ might intersect two other (than $e_4$) edges of $F$, this may create a path of length three in the set $\eta_{M,M'}(I,F)$. This scenario can repeat many times and, consequently, $\eta_{M,M'}(I,F)$ may contain many (disjoint) paths of length three. Thus, the image of $\eta_{M,M'}$ would be much larger than $\Omega(H)$, not yielding (\ref{cut}).

     Another direction of further research is to try to obtain an FPRAS for perfect matchings in \emph{dense} $k$-uniform hypergraphs, where the density is measured as, e.g.,  in \cite{KRS}. For $k=2$ this was done in \cite{js}. The 3-combs are an obstacle here too, but in addition, we are facing the problem of the necessity of including into the state space of the Markov chain matchings much smaller than the perfect ones (in \cite{js} the state space consisted only of perfect and near-perfect matchings, that is, matchings missing just two vertices).

\section*{Acknowledgements}

   We thank Martin Dyer, Mark Jerrum and Alex Samorodinsky
    for a number of stimulating discussions.



\end{document}